\documentclass[11pt]{article}
\usepackage[round]{natbib}
\usepackage{authblk}
\usepackage{pdfpages}
\usepackage{amsmath,amssymb,amsfonts,mathrsfs,amsthm}
\usepackage{bbm}
\usepackage{epsf,epsfig,setspace,multirow,subfigure,comment,gensymb,geometry}
\usepackage{booktabs,tabularx,array,xcolor}
\usepackage{hyperref}
\usepackage{tikz}
\usetikzlibrary{arrows.meta,positioning,fit,calc}
\newcolumntype{Y}{>{\raggedright\arraybackslash}X}
\newtheorem{definition}{Definition}
\newtheorem{proposition}{Proposition}
\newtheorem{theorem}{Theorem}
\newtheorem{lemma}{Lemma}

\providecommand{\proofsubparagraph}{\paragraph}
\title{Competing Auctions in Intermediated Markets}
\author[1]{Bruno Mazorra}
\author[1,2]{Minghao Pan}
\author[1]{Christoph Schlegel}
\affil[1]{Flashbots}
\affil[2]{California Institute of Technology}
\newif\iflongversion
\longversiontrue
\begin{document}
\maketitle
\begin{abstract}
We analyze competing auctions in intermediated markets, where a seller selects among parallel mechanisms for the sale of a single good, most prominently the relay-and-protocol architecture of proposer-builder separation in Ethereum.

When the intermediary can enforce single-homing on its bidders, sealed-bid second-price intermediary auctions fully unravel into the sealed first-price principal auction; open bidding-format intermediaries unravel only partially, collapsing into first-price in equilibrium under symmetric latency and sorting fast bidders to the intermediary under asymmetric latency. Any last-look advantage is removed through the availability of a credible sealed bidding channel. These results extend to multi-plexing environments (no enforcement by the intermediary). For second-price relays, the multi-homing result covers all symmetric Bayes--Nash equilibria, including mixed strategies.

While the unraveling result indicates that the availability of a sealed first-price bidding channel pushes the overall market to the same auction structure, the very assumption of the credibility of such a channel is problematic, as the seller may have an incentive to leak information: a first-price auction is leakage-resistant in the presence of a single ``fast'' bidder but not against two or more. However, if the seller can credibly commit to not leak bids, it is optimal for them to do so. 

A main motivation is the forthcoming Glamsterdam update of Ethereum: our analysis suggests that the availability of an in-protocol (first-price) bidding channel severely limits the design space for out-of-protocol auctions by relays and other intermediaries.
\end{abstract}
\section{Introduction}
In many markets, the same good can be bought through different channels, often through competing auction mechanisms. In classical online marketplaces (eBay etc.) competing auctions usually mean that there are competing sellers of the same type of good, and the sellers themselves may multi-home, offering their goods or services across different platforms at the same time (e.g. in ride-sharing services). In other contexts, however, there may be a single seller of a single good, while competition happens between intermediaries. Competing auctions by competing sellers have been studied in a rich literature starting with \cite{McAfee1993}; competing auctions run by intermediaries for a single good have received much less attention.

This is not only a theoretical gap, but practically very relevant. The most familiar example from traditional finance is order routing: a broker with a single order to execute chooses among competing venues running different auction mechanisms (lit exchanges, dark pools, retail wholesalers, periodic auctions, RFQ systems) and is required under best-execution rules to route the order to whichever venue offers the best outcome. Blockchain systems, which will be our main motivating examples, are even more abundant in intermediation services, and we observe competing auctions throughout the transaction inclusion pipeline. Ethereum's MEV ecosystem offers a particularly rich set of examples, from order flow auctions (MEV-protect, MEV-blocker, etc.) over routing and RFQ services (1inch, CowSwap etc.) to builder competition (Titan, BuilderNet etc.) for searcher bundles. In each case, the transaction or intent is the single good, and competing intermediaries run different auction mechanisms to capture it.

A forthcoming Ethereum protocol upgrade, enshrined Proposer-Builder Separation \iflongversion(ePBS, \citealp{eip7732})\else(ePBS)~\cite{eip7732}\fi, introduces another instance of competing auction mechanisms. Currently, Proposer-Builder Separation, the market structure through which block builders bid for their blocks to be proposed, is implemented out-of-protocol with the help of intermediaries (so-called relays). Under ePBS, the protocol provides a native bid-and-commitment layer for builder bids, while off-protocol channels may still coexist with it: bids may be routed through public protocol gossip, private proposer-facing builder RPCs, or new relay-like intermediaries. Thus, we will possibly observe competition between direct in-protocol and intermediated bids.

In this paper, we study this competition between an in-protocol bidding channel and intermediated ones. Our analysis predicts that the introduction of an in-protocol sealed-bid first-price channel severely constrains the design space available to off-protocol intermediaries, and we characterize which relay formats can and cannot survive this competition.

We analyze the question of what an in-protocol bidding channel does to relay intermediation from several angles: first we study what kind of pricing rule the competing auction can maintain in the presence of the in-protocol outside option. A plausible alternative choice to a first-price auction (sealed or open in the form of an English auction) would be to run a second-price auction format. A second-price auction has the advantage that bidders have an incentive to bid their value truthfully in it.\footnote{Currently, some relays like the Ultrasound relay experiment with ``second-price'' features.} In the presence of a competing first-price channel, however, bidders still bid truthfully \emph{conditional on entering} the second-price auction (Lemma~\ref{lem:sp_truthful}), but entering it at all becomes unattractive: in the single-homing case, the second-price auction unravels and bidders self-select into the first-price (in-protocol) auction (Theorem~\ref{SPUnraveling}). In the multi-homing case, every symmetric Bayes--Nash equilibrium, including mixed-strategy equilibria, has the standard first-price allocation and payment~(Theorem~\ref{Essential}).

Second, we study how information can be shared throughout the auction: open bidding as in the status quo is naturally undermined by the availability of a sealed-bid in-protocol channel. Bidders have an incentive to listen to an open-bid feed but to bid through a sealed channel if available. However, the seller (or an intermediary acting on the seller's behalf) can choose to share bids received during the auction period, orchestrating a (partially) open-bid auction. Such information sharing naturally unravels in the sense that bidders have an incentive to bid as late as possible to avoid having their bids shared, effectively turning the auction into a sealed-bid auction (Propositions~\ref{nolastlook} and~\ref{cor:open_relay_no_leakage}). However, in the presence of latency differences by bidders, information sharing (of the slower bidders' bids with the faster bidders) may be feasible. We analyze this in a model where a subset of bidders is ``fast'' and the seller or intermediary may share bids with these bidders. We find that sharing information with a single fast bidder is not optimal for the seller (Proposition~\ref{1leak}), whereas sharing information with two or more fast bidders (in case there are that many fast bidders) is a best response for the seller (Theorem~\ref{2leak}) and fast bidders would benefit from bid disclosure (Proposition~\ref{lastlookrents}). However, if the seller can credibly commit ex-ante to a disclosure rule, the optimal commitment is to share no information at all (Proposition~\ref{prop:commitment_show}).

The interpretation of the ``information disclosure'' results in the Ethereum context is straightforward: in the absence of a credible first-price sealed-bidding channel, if there are at least two ``fast'' bidders, a proposer would have an incentive to leak the bids it has received (in or out-of protocol), effectively making the in-protocol bidding channel leaky.  Practically, this could be orchestrated through a service where proposers forward the received bids to a relay that integrates them in their bid stream. However, there are potentially ways to establish a credible sealed bid first-price channel: proposers could use reputation (in the case of large staking pools like Lido) or hardware solutions such as Trusted Execution Environments (TEEs) to commit to not leak bids that they have received. Alternatively, they can delegate to an intermediary, i.e. a potential ``sealed-bid first-price'' relay. As commitment to non-disclosure is optimal if possible, such a solution can quite plausibly recover the no-disclosure equilibrium institutionally.

We proceed as follows: in the next two subsections, we provide background on the literature on competing auctions, and institutional background on Ethereum's blockbuilding stack. In Section~\ref{sec:Modelling} we discuss our modeling choices for the subsequent analysis. The discussion is more detailed than the usual scope of such a section, primarily because some modeling choices are very much informed by institutional details of the block building ecosystem in Ethereum which we want to discuss in this context. In Section~\ref{sec:results} we derive our results, which are summarized in Table~\ref{tab:summary-results}. 
In Section~\ref{sec:conclusion}, we conclude with practical implications of our analysis and open questions.
\begin{table}[!t]
\centering
\small
\setlength{\tabcolsep}{4pt}
\renewcommand{\arraystretch}{1.15}
\caption{Summary of results.}
\label{tab:summary-results}
\begin{tabularx}{\textwidth}{@{}p{1.6cm}p{4.4cm}Y@{}}
\toprule
Section & Setting & Result \\
\midrule
\S~\ref{sec:single_homing}
&
Single-homing, sealed second-price relay
&
The native first-price auction is the unique symmetric equilibrium outcome (Theorem~\ref{SPUnraveling}).
\\
\S~\ref{sec:multi_homing}
&
Multi-plexing, second-price relay
&
Every symmetric equilibrium has the standard first-price allocation and payment (Theorem~\ref{Essential}).
\\
\S~\ref{sec:english_without_last_look}
&
No last-look advantage
&
All-first-price is an equilibrium; the all-English profile is unstable in large markets (Proposition~\ref{nolastlook}).
\\
\S~\ref{sec:last_look_no_commitment}
&
One fast bidder
&
The first-price bidding channel is leakage-resistant (Proposition~\ref{1leak}).
\\
\S~\ref{sec:last_look_no_commitment}
&
Multiple fast bidders, $|S|\geq 2$
&
The proposer leaks in equilibrium (Theorem~\ref{2leak}).
\\
\S~\ref{sec:last_look_commitment}
&
Last-look with credible proposer commitment
&
The proposer weakly prefers to commit to a credible sealed-bid first-price channel (Proposition~\ref{prop:commitment_show}).
\\
\bottomrule
\end{tabularx}
\end{table}

\subsection{Competing Auctions \& Related Literature}
The literature on \emph{competing auctions}~\cite[]{McAfee1993,PetersSeverinov1997,BurguetSakovics1999,ellison2004competing} studies the constellation where multiple sellers each offer their own good through a chosen mechanism, with buyers single-homing across sellers.  Our setting differs structurally: a single seller has a single good and faces multiple parallel mechanisms run by intermediaries, with the winning mechanism determined ex-post by a max-revenue rule. The strategic considerations are therefore quite different: bidders are not choosing between substitute goods but between different bidding channels, and competition operates on intermediaries rather than on sellers.

The literature on \emph{platform competition} \cite[]{CaillaudJullien2003,Armstrong2006,RochetTirole2006} introduces the distinction between single-homing and multi-homing that also plays a role in our analysis. The general finding - buyer-side multi-homing intensifies competition between platforms, while single-homing produces winner-take-all dynamics - has no analog in our setting which focuses on equilibrium auction outcomes rather than platform pricing.

Our information leakage results connect to the literature on \emph{seller-controlled information disclosure in auctions}. The classical result here is the linkage principle of~\iflongversion\cite{MilgromWeber1982}\else Milgrom and Weber~\cite{MilgromWeber1982}\fi, which establishes that under affiliated and interdependent values, the seller benefits from committing to reveal information that links the price more strongly to bidders' signals. This logic does not apply in our setting: at the builder level, values are private (although possibly drawn from a correlated distribution) and known with certainty, so additional information about other bidders' valuations has no effect on a bidder's own valuation. The relevant disclosure question becomes strategic rather than informational: the seller's choice is not about value discovery, but about changing the strategic environment by giving some bidders the ability to react to others' bids. Related issues are studied in the literature on information design in auctions under private values~\cite[]{BergemannPesendorfer2007}. 

The closest non-blockchain analog of our setting is in \emph{market microstructure}, where a substantial literature studies how a single trade is routed across competing exchanges, dark pools, and wholesalers under best-execution obligations~\cite[]{FoucaultPaganoRoell2023, OHaraYe2011}. A close institutional parallel to our latency-advantaged-bidder setting is the case of last-look provision in FX markets~\cite[]{Oomen2017}, where dealers selectively respond to client quote requests based on observed market information.

\subsection{From MEV-Boost to ePBS}
\label{sec:from-mevboost-to-epbs}

Ethereum's current block-building market is organized through MEV-Boost, an off-chain implementation of proposer-builder separation. Validators outsource execution-payload construction to specialized builders. Builders compete by submitting blocks and associated bids to relays, and the proposer signs the highest-paying blinded block header returned by MEV-Boost. The relay is therefore the central intermediary: it runs the auction, hides the block contents until the proposer has committed, and then releases the full payload. In the standard, non-optimistic flow, the relay also verifies block validity and proposer payment before forwarding the bid. The builder sets itself as the block's \texttt{feeRecipient} and includes a final transfer to the proposer's registered \texttt{feeRecipient} address; the relay checks that this transfer matches the advertised bid.

Relay designs have evolved to reduce latency on the critical path. In optimistic relay designs, builders post collateral, allowing the relay to accept bids before completing full validation. The relay still intermediates the auction and forwards the blinded header to the proposer, but payment and validity checks are performed ex-post. If the winning block is invalid or fails to pay the promised amount, the builder's collateral is used to compensate the proposer. Optimistic relaying therefore makes the MEV-Boost auction more latency-efficient, since builders can submit competitive bids later in the slot, but it does not remove the relay as the trusted off-chain intermediary.

Thus, MEV-Boost implements PBS without changing Ethereum consensus, but it does so by relying on off-chain infrastructure. Relays solve the fair-exchange problem between proposers and builders: builders need protection against payload theft or payload leakage before commitment, while proposers need assurance that the promised block is valid and pays the advertised amount. The resulting market is naturally modeled as an open, continuous, relay-mediated auction. Builders can update bids during the slot, relays intermediate the information flow, and latency determines how late bids can arrive and still be competitive.

Enshrined proposer-builder separation, or ePBS, changes this architecture by moving the core commitment-and-payment logic into the protocol. In the current EIP-7732 design \cite[]{eip7732}, the beacon block no longer contains the full execution payload on the attestation-critical path. Instead, it contains a signed builder commitment, the \texttt{SignedExecutionPayloadBid}, which commits to a later execution payload and specifies the value to be paid to the proposer. Builders become protocol-recognized entities with beacon-chain balances, so the committed payment can be deducted from the builder's balance and later credited to the proposer. Timely reveal is checked by a Payload Timeliness Committee (PTC), while full execution validation is deferred to a later point in the slot. The current EIP-7732 design has some potential flaws if payload-delivery failure is not sanctioned and the PTC deadlines are positioned too far in the slot~\cite[]{mazorra2025free}. As there are plausible mitigation mechanisms for this problem, we assume throughout the analysis that these ``free-option'' considerations do not play a role for the questions we seek to answer in this paper.

These protocol changes also alter the market structure. Under MEV-Boost, the relay observes and intermediates the bidding process, and the auction is effectively open, continuous, and latency-sensitive. Under ePBS, the protocol provides a native bid-and-commitment layer, while off-protocol channels may still coexist with it. The post-ePBS market may therefore contain competing auction formats: a native in-protocol mechanism and external auction channels that can differ in observability, timing, trust assumptions, and payment guarantees. Our goal is to understand when these mechanisms coexist, when one unravels into the other, and how the change from trusted relay-mediated exchange to protocol-mediated exchange affects builder and proposer incentives.

\subparagraph*{Trustless vs Trusted Payments.} In MEV-Boost, payments are mediated by trusted off-chain infrastructure: the relay verifies that the builder's block contains the promised payment and releases the payload only after the proposer has signed the blinded block. In ePBS, by contrast, payment can be enforced in protocol. A builder's signed bid commits to a payment amount, which can be deducted from the builder's beacon-chain balance and later credited to the proposer.

The main cost of trustless payments is that builders must pre-fund their bids. This can increase entry barriers, since a builder needs sufficient in-protocol balance to participate. For ordinary blocks, this constraint may be mild: most bids are small, and if deposit and withdrawal frictions are low, the required working capital need not be large. For high-value blocks, however, the balance constraint can become binding. A builder with high value but insufficient balance may be unable to express its full willingness to pay through the trustless mechanism. In such cases, bids may still be routed through trusted relationships, relays, or other intermediaries that can provide credit, guarantees, or faster settlement. Subsequently, we will not model this particular friction explicitly, but the reader should keep in mind that inter-mediated bids may still have a role and certain advantages over in-protocol bidding, even if trustless bidding would be optimal based on ``pure auction considerations'' that are the main focus of our analysis.

\section{Modeling Consideration}\label{sec:Modelling}

\subsection{Private Values}
We will work throughout the paper in a private value setting. Some justification is needed here, as this often seems to be a point of possible confusion. Most block value is generated from financial transactions, which inherently have a common-value component (with at least some uncertainty and heterogeneous signal on this value). However, this does not mean that block builders have a common value for proposing the next block. An uncertain common value is effectively realized up-stream at the searcher level. Block builders receive bundles and compose them into a block whose value for the builder is deterministic and known with certainty. This would be different if builders did searching themselves, a market dynamic called ``searcher-builder integration'' that has been discussed extensively in the past~\cite[]{pai2024structural}. However, this dynamic has never materialized in reality, and at the current point in time the dominant Ethereum builders (Titan, BuilderNet, Quasar) are neutral and do not search. In particular, ``price discovery'' in the PBS auction, sometimes cited as an advantage of the current open auction format, is - if at all - relevant for searchers to update their bids to the builders.
Thus, our analysis focuses on the practically relevant case of private values. As different builders tend to receive similar bundles (or even the same bundles from the same searcher), these private values are naturally positively correlated. However, we refrain from modeling the correlation structure explicitly, for the sake of analytical simplicity. We note, however, that positive correlation generally pushes in the same direction as the logic of most of our results, so that they likely generalize to models of positively  correlated private valuations.\footnote{A trivial form of introducing correlation for which our result would hold without modification of the proofs is through a block-level public state: Fix a block
and condition on its realized public state; let $F$ denote
the value distribution induced by that state. Assume there is  a family $\mathcal F$ of distributions
supported on the same compact interval. For the realized distribution $F\in\mathcal F$,
bidders' values satisfy
\[
    v_i\mid F\sim_{\mathrm{i.i.d.}}F,
    \qquad i=1,\ldots,n .
\]Thus values can be correlated before conditioning:
the public state can shift the distribution faced by all bidders in the block.
Economically, $F$ captures the public market conditions of the current block.
Whenever a proposition refers to $F$, it can be read as conditional on an
arbitrary realized $F\in\mathcal F$.}
Thus, throughout the paper we assume that valuations are drawn i.i.d. from a distribution $F$ that is absolutely continuous and strictly increasing on a compact interval
$[\underline v,\overline v]$. We denote its density by $f$.

\subsection{Single vs Multi-plexing of Bids}
The in-protocol auction in ePBS is permissionless. The protocol can neither restrict who enters the auction nor restrict whether they submit multiple bids through different channels (in-protocol and directly to the proposer or through a relay) or just a single bid. Thus, the ePBS auction is by default ``multi-plexing'' friendly. The case of relays is more subtle. Relays already today restrict certain functionalities (optimistic relaying) to registered and staked participants while being generally open to all participants for their standard service. In a future with ePBS, it is still to be seen whether relays become ``more permissioned'' than they are today. One potential reason they could develop in this direction is that they want to offer additional services to remain attractive even with an in-protocol mechanism in place; for example, block-merging by the relay is easier to coordinate with a permissioned set of builders. More importantly for our context, a relay might want to sustain an auction format that is otherwise not sustainable in the multi-plexing case. The current open-bidding PBS auction, for example, is hard to maintain in an ePBS world: if a relay offers an open-bidding mechanism, then bidders have a natural incentive to listen to the bidding feed of the relay but to place their own bid in a sealed way directly to the proposer, with the relay auction unraveling as a consequence. If a relay wants to fight this logic, it might restrict bidding to a permissioned set of builders and impose that, in case a builder bids with them in a particular auction round, they need to commit to not bidding through an external channel in the same round. This assumes, of course, that deviations are detectable with sufficient probability by the relay. If a relay manages to implement single-homing within its service, it makes the entire market single-homing: bidders need to select into one of the two mechanisms before they bid.

We find both scenarios - single-homing enforced at the relay level and multi-homing - plausible at first view. Thus, we think it is valuable to analyze both cases.

\subsection{Latency and Last-Look Advantage}
The performance and equilibrium analysis of  auctions in time-sensitive environments may be influenced by bidder latency considerations. 
The current MEV-Boost auction, for example, is naturally modeled, at an idealized level, as an English or candle-style auction. Relays receive continuously updated bids during the slot, builders can react to the current competitive level, and the process ends only when the proposer must commit to a block. 
But in the current MEV-Boost world, bid streams and timing frictions make latency economically meaningful: a builder who observes a competing bid and can still respond before the deadline while others cannot obtain a form of \emph{last look}~\cite[]{pai2024structural} advantage.

For our analysis of open bidding and bid disclosure we, therefore, consider two types of environments:
We first consider the clean benchmark in which no bidder has a last-look advantage. In that benchmark, the relay auction may reveal information internally among relay participants, but bidders can submit a last bid that is not observed by competitors. Moreover, the relay winner does not observe the native first-price bids before choosing the proposer-facing bid and vice versa.

The no-last-look benchmark is an idealization.  We then study the more interesting case of last look explicitly. We treat technological latency as exogenous, but the exploitation of latency as endogenous: low-latency builders can use their advantage only if some information is made available before the deadline. That information may be public, relay-mediated, or selectively disclosed by the seller.

This leads to two versions of the last-look problem. In the first, the seller cannot credibly commit to a disclosure strategy. After receiving early bids, the seller may decide history by history which private bids to leak and to whom (of the latency advantaged bidders).\footnote{\label{fn:no_commitment}For the PBS application, practically, the proposer will still decide ex-ante, before the auction, how to configure their node (or more likely which third-party side-car software to install). The point is that while the proposer chooses their strategy ex-ante in this application, they cannot credibly commit to a disclosure strategy that is not an equilibrium strategy.} In the second, the seller can commit ex-ante to a specific disclosure rule. Practically, this disclosure can be implemented through a trusted intermediary, such as a relay in the PBS case. The distinction matters because a sealed-bid procedure can be ex-ante attractive, while after the early bids have arrived the seller may have an incentive to leak them to latency-advantaged builders.

\section{Results}\label{sec:results}
\iflongversion
Proofs omitted from the main text are provided in Appendix~\ref{app:proofs}.
\else
Proofs omitted for space are provided in the
\href{https://arxiv.org/abs/2606.06633}{full version}.
\fi
\subsection{Single-Homing: First-Price vs Second-Price Auction}
\label{sec:single_homing}
We first analyze the single-homing case.  There are $n\geq 2$ bidders who bid for the right to decide on the content of the next block. The market structure is as follows: there are two competing auctions. The first auction is a first-price sealed-bid auction (the native ePBS auction) and the second a competing relay auction. The bidders can submit a bid to one of the two auctions or neither auction. The two auctions are run on the submitted bid. Whichever auction yields higher revenue allocates the item and charges the winner. 

Bidders choose ad interim, i.e. after they know their valuation but before they send their first bid, through which channel they bid in the auction. The natural interpretation is that of an intermediary (relay) that can detect deviations (with high probability) where bidders bid through both channels and punish this ``multi-plexing'' behavior.

We first study whether a competing auction can enforce a second-price pricing rule. We restrict attention to symmetric equilibria.\footnote{A natural question is whether the unraveling results in this and the following section follow more directly from revenue equivalence. Revenue equivalence determines the expected payment in any efficient symmetric mechanism with zero surplus at the lowest type, but it presupposes a fixed allocation rule. In our setting, the allocation rule is itself an equilibrium object of interest. In principle, the good can be allocated inefficiently - for example, if a high-value bidder selects themselves into the second-price auction but is outbid by a lower-value bidder in the first-price auction. We therefore cannot invoke revenue equivalence to short-cut the analysis.}

\begin{theorem}[Single-homing unraveling]\label{SPUnraveling}
If bidders must choose to participate in exactly one
of the first-price auction and the sealed second-price auction, then, up to
$F$-null sets, there is a unique symmetric Bayes--Nash equilibrium. In this
equilibrium, every type chooses the first-price auction and bids according to
the standard symmetric first-price equilibrium.
\end{theorem}

We separate the proof into several lemmas.

\begin{lemma}\label{lem:sp_truthful}
In equilibrium, if an agent bids in the second-price auction, then she bids truthfully almost surely.
\end{lemma}
\begin{proof}
By symmetry, we may assume that bidders who participate in the second-price auction bid truthfully almost surely. The reason is simple. Truthful bidding weakly dominates every other bid in the second-price auction: lowering the bid below one's value can only remove wins at profitable prices, while raising the bid above one's value can only add wins at loss-making prices. Moreover, in a symmetric equilibrium any positive-measure set of non-truthful second-price bids would make this domination strict on a positive-probability event. If a positive measure of types overbid, then by symmetry there is positive probability that all second-price participants have values below some cutoff but bids above it; the winner then pays more than her value, while truthful bidding would avoid this loss. If a positive measure of types underbid, then with positive probability all second-price participants have values above some cutoff but bids below it; at least one bidder then loses even though a truthful bid would win at a price below her value. Thus any non-truthful second-price bid used on a positive-measure set would admit a strictly profitable deviation to the truthful bid, contradicting equilibrium. Hence, conditional on participating in the second-price auction, bidders bid truthfully almost surely.
\end{proof}

\begin{lemma}\label{lem:fp_bid_monotone_single_homing}
If two bidders with types $v,w\in(\underline{v},\overline{v}]$ participate in the first-price auction with positive probability, then $v>w$ implies $b(v)>b(w)$.
\end{lemma}

\begin{proof}[Proof of the theorem]
For notational simplicity, normalize the support to $[0,1]$. Let $\mu$ be the distribution associated with $F$. Suppose by contradiction that there exists a symmetric equilibrium with positive participation in the second-price auction. Let $E\subseteq[0,1]$ denote the set of types who participate in the second-price auction. 

We first show that $\alpha:=\operatorname*{ess\,inf} E>0$. A type with valuation $v$, by Lemmas~\ref{lem:sp_truthful} and~\ref{lem:fp_bid_monotone_single_homing}, obtains from the second-price auction payoff
\begin{equation*}
    U_2(v)
    :=
    \int_{[0,v]\cap E}
        (v-t)
        \sum_{k=1}^{n-1}
        \binom{n-1}{k}
        k\cdot
        \mu(b^{-1}([0,t])\cap E^c)^{n-1-k}
        \mu([0,t]\cap E)^{k-1}
    \,d\mu (t).
\end{equation*}
Here $\binom{n-1}{k}$ samples $k$ other participants of the second-price auction, and the following factor $k$ selects one participant to have value $t$.

Suppose that the function
\begin{equation*}
    t
    \mapsto
    (v-t)
    \sum_{k=1}^{n-1}
    \binom{n-1}{k}
    k\cdot
    \mu(b^{-1}([0,t])\cap E^c)^{n-1-k}
    \mu([0,t]\cap E)^{k-1}
\end{equation*}
is maximized at some $t=t^*$. Then
\begin{equation*}
    U_2(v)
    \leq
    \mu([0,v]\cap E)(v-t^*)
    \sum_{k=1}^{n-1}
    \binom{n-1}{k}
    k
    \mu([0,t^*]\cap E)^{k-1}
    \mu(b^{-1}([0,t^*])\cap E^c)^{n-1-k}.
    \tag{1}
\end{equation*}

If the type chooses to bid $t^*$ in the first-price auction instead, her utility payoff is at least
\begin{equation*}
\begin{aligned}
    &(v-t^*)\mu(b^{-1}([0,t^*])\cap E^c)^{n-1} \\
    &\qquad
    +
    (v-t^*)
    \sum_{k=1}^{n-1}
    \binom{n-1}{k}
    \mu(E)
    \mu([0,t^*]\cap E)^{k-1}
    \mu(b^{-1}([0,t^*])\cap E^c)^{n-1-k}.
\end{aligned}
    \tag{2}
\end{equation*}
Here $\binom{n-1}{k}\mu(E)\mu([0,t^*]\cap E)^{k-1}$ comes from first sampling $k$ participants in the second-price auction, letting one participant bid arbitrarily, and having the remaining $k-1$ participants bid less than $t^*$.

By comparing $(2)$ with $(1)$ term by term, as long as
\begin{equation*}
    (n-1)\mu([0,v]\cap E)<\mu(E),
\end{equation*}
we have $(2)>(1)$. This is true when $v$ is small enough. In other words, when $v$ is small enough, $v$ prefers the first-price auction to the second-price auction, and so $\alpha>0$.

Let $\varepsilon>0$ be small such that $\alpha+\varepsilon\in E$. When participating in the second-price auction, bidder $\alpha+\varepsilon$ has a chance of winning bounded above by
\begin{equation*}
    (n-1)\mu([\alpha,\alpha+\varepsilon]).
\end{equation*}
So her utility payoff is at most
\begin{equation*}
    \varepsilon(n-1)\mu([\alpha,\alpha+\varepsilon]).
\end{equation*}
However, if she bids $\alpha$ in the first-price auction, she wins whenever all other participants have value smaller than $\alpha$, so her utility payoff is at least
\begin{equation*}
    \varepsilon \mu([0,\alpha])^{n-1}.
\end{equation*}
Thus, taking $\varepsilon$ sufficiently small,
\begin{equation*}
    \mu([0,\alpha])^{n-1}
    >
    (n-1)\mu([\alpha,\alpha+\varepsilon]),
\end{equation*}
and so, bidder $\alpha+\varepsilon$ prefers the first-price auction, a contradiction.
\end{proof}
Theorem~\ref{SPUnraveling} implies that, under single-homing, the sealed
second-price relay does not attract participation: in the unique symmetric
equilibrium, all bidders bid in the first-price auction, and the item is
allocated through that channel. In the Ethereum interpretation, the block is
therefore allocated through the native direct-communication channel rather than
through the relay. This conclusion, however, relies on the relay's ability to
enforce single-homing. In a permissionless environment, a relay may restrict
the identities that participate in its auction, but it cannot generally verify
that an identity corresponds to a unique economic bidder. A bidder who wants to
bid in both channels can use Sybil identities, submitting one bid through the
relay and another through the native auction while appearing as distinct
participants. Thus, single-homing has bite only to the extent that such
deviations can be detected. When Sybil participation is feasible, bidders can
circumvent the restriction by splitting bids across identities, so the effective
environment is one of multi-homing. We analyze that case next.

\subsection{Multi-Homing: First-Price vs Second-Price Auction}
\label{sec:multi_homing}

The essential uniqueness result holds without restrictions on the shape of
bidding strategies. Each bidder may submit at most one nonnegative bid in
each auction; zero also represents abstention. Values lie in
$[\underline v,\overline v]\subseteq[0,\infty)$. The first-price auction is
selected when the two auctions yield equal revenue.

We write $v\mapsto(B(v),W(v))$ for a possibly mixed bidding strategy,
where $(B(v),W(v))$ is a random vector of first-price and second-price
bids whose distribution depends measurably on $v$. A pure strategy is the
special case in which this vector is deterministic. Private randomization
is independent across bidders.

\begin{theorem}[Essential uniqueness]\label{Essential}
Let $v\mapsto(B(v),W(v))$ be a possibly mixed bidding strategy inducing a
symmetric Bayes--Nash equilibrium of the multi-homing game. For
$F$-almost every $v$,
\[
    B(v)=b^{\mathrm{FP}}(v),\qquad W(v)\leq b^{\mathrm{FP}}(v)
    \quad\text{almost surely},
\]
where
\[
    b^{\mathrm{FP}}(v)
    =v-\frac{\displaystyle\int_{\underline v}^{v}F(t)^{n-1}\,dt}
                {F(v)^{n-1}},\qquad v>\underline v.
\]
Consequently, the good is allocated through the first-price auction almost
surely, and the allocation and payment coincide with those of the standard
symmetric first-price equilibrium. The pure strategy
$v\mapsto(b^{\mathrm{FP}}(v),0)$ induces a symmetric equilibrium.
\end{theorem}
Call a bid pair $(b,w)$ \emph{second-price-active} if $w>b$,
and let $\mathcal A=\{W>B\}$ denote the event that a bidder chooses
such a pair.
The argument compares each second-price-active action with all first-price-only
deviations and then integrates these inequalities across second-price-active
actions. This avoids identifying the type of a marginal first-price bidder.

For the proof, fix an equilibrium strategy $v\mapsto(B(v),W(v))$ as in
Theorem~\ref{Essential}. Let $V\sim F$ be a bidder's random value, and
abbreviate her induced bids $B(V)$ and $W(V)$ by $B$ and $W$.
All probabilities below are induced by the value distribution and the
fixed strategy, and statements about realized values and bids are
understood up to null sets under this law. 

\begin{lemma}\label{lem:mh_canonical}
The equilibrium bids satisfy
\[
    W\leq B<V\quad\text{or}\quad W=V>B
    \qquad\text{almost surely}.
\]
Moreover, the subdistribution of first-price bids from actions satisfying
$W\leq B$ has no atoms.
\end{lemma}
\begin{proof}
A bidder with value $v$ who chooses a bid pair $(b,w)$ with
$b>v$ and wins, pays at least $b>v$,
regardless of which auction is selected. Since abstention guarantees
zero utility, such a bid pair can be used in equilibrium only if it
loses with probability one. Suppose that $\Pr(B>V)>0$. There exist $a<c$ such that 
$\Pr(V\leq a<c\leq B)>0$. In particular, by independence, with
positive probability all bidders satisfy this event. On this event, the
winning bidder pays at least $c$ despite having value at most $a$,
a contradiction. Thus $B\leq V$ almost surely.

A bidder with value $v$ who chooses a bid pair $(b,w)$ with
$b\leq v$
 can win through the second-price
auction only at a price above $b$. Thus, bid pair $(b,w)$ is weakly dominated by $(b,v)$: if $w<v$, deviating to $(b,v)$ adds
wins at profitable prices, if $w>v$, deviating to $(b,v)$ only removes
wins at loss-making prices.
The dominance is strict whenever non-truthful second-price-active
actions occur with positive probability. Suppose first that
$\Pr(B<W<V)>0$. There
are constants $a<c$ such that
$ \Pr(B<a<W<c<V)>0.
$
If all bidders satisfy these inequalities, the second-price auction
is selected: its second-highest bid exceeds $a$, whereas every
first-price bid is below $a$. Any losing bidder can then profitably
raise her second-price bid to her value, which exceeds every other
second-price bid.

Similarly, if $\Pr(B\leq V<W)>0$, there is a constant $a$ such that
$\Pr(B\leq V<a<W)>0$. If all bidders satisfy these inequalities,
the second-price auction is selected and its winner pays more than
her value, which truthful bidding avoids.
Both possibilities contradict equilibrium. Thus $W=V$ on $\{W>B\}$.

It follows that both coordinates are at most the bidder's value. Every
type $v>\underline v$ can consequently obtain a strictly positive payoff
by bidding $(x,0)$ with $\underline v<x<v$: she wins whenever every
opponent's value is below $x$. A bid pair $(b,w)$ with $b\geq v$ cannot yield a
positive payoff for type $v$, so $B<V$ almost surely.

Finally, suppose that first-price bids from actions with $W\leq B$ had an
atom at $x$. These bidders have values strictly above $x$. On the event
that all opponents also use such an action at $x$, a bidder faces a
first-price tie. Raising her bid to $x+\varepsilon$ and setting her second-price
bid to zero removes this tie. Outside this event, it cannot remove any
win and increases the payment on a retained win by at most $\varepsilon$.
For some positive-measure subset of the atom the margin $v-x$ is bounded
away from zero. For sufficiently small $\varepsilon$, a bidder in this
subset who does not receive priority in every tie has a profitable
deviation. Under symmetric tie-breaking every bidder has this property;
under a fixed priority rule symmetry of strategies gives the same
contradiction for a lower-priority bidder. Hence the subdistribution has
no atoms.
\end{proof}

\begin{proof}[Proof of Theorem~\ref{Essential}]
Write $m=n-1$. By Lemma~\ref{lem:mh_canonical}, a second-price-active
action has the form $(b,v)$ with $b<v$. Define the subdistribution functions
\begin{align*}
    K(x):=\Pr[\mathcal A,\ V\leq x],\quad
    L(x):=\Pr[\mathcal A,\ B\leq x],\quad
    J(x):=\Pr[\mathcal A^c,\ B\leq x],
\end{align*}
and set
\[
    H(x):=J(x)+K(x),\qquad
    R(x):=L(x)-K(x),\qquad
    D(x):=H(x)+mR(x).
\]
The functions $J$ and $K$ are continuous, so $H$ is continuous and
nondecreasing. Since $B<V$ on $\mathcal A$,
\[
    H(x)=\Pr[B\leq x,\ W\leq x],\qquad
    R(x)=\Pr[\mathcal A,\ B\leq x<V]\geq0.
\]
The function $R$ has bounded variation; its jumps, if any, are upward
jumps inherited from $L$. It is zero below zero and at or above
$\overline v$. Integrals against $dL$ or $dR$ use right-hand values at
jumps. Ratios with both numerator and denominator zero are set to zero.

Let $U(v)>0$ denote equilibrium utility for $v>\underline v$.
A first-price-only bid just above $x$ wins whenever all opponents have
both coordinates at most $x$, or exactly one opponent has
$B\leq x<W$ and all remaining opponents have both coordinates at most
$x$. Taking the bid down to $x$ gives
\begin{equation}\label{eq:mh_fp_deviation}
\begin{split}
    U(v)&\geq(v-x)\bigl[H(x)^m+mH(x)^{m-1}R(x)\bigr]\\
        &=(v-x)H(x)^{m-1}D(x),\qquad 0\leq x<v.
\end{split}
\end{equation}
This limiting deviation handles atoms in the first-price bids of
second-price-active actions. When $m=1$, the factor $H^0$ is interpreted as one.

For an equilibrium second-price-active action $(b,v)$, the exact payoff is
\begin{equation}\label{eq:mh_active_payoff}
    U(v)=(v-b)H(b)^m
       +m\int_{(b,v)}(v-t)H(t)^{m-1}\,dK(t).
\end{equation}
The first term is the first-price win. In the second term, one opponent
uses a second-price-active action with value $t\in(b,v)$ and all other
opponents have both coordinates below $t$. This opponent determines the
payment in the second-price auction. An inactive opponent cannot determine
a second-price payment strictly above all first-price bids, because her
second-price bid is at most her first-price bid.

Divide \eqref{eq:mh_active_payoff} by $U(v)$ and apply
\eqref{eq:mh_fp_deviation} separately to its first term and to each
integrand. We obtain
\[
    1\leq \frac{H(b)}{D(b)}
        +m\int_{(b,v)}\frac{dK(t)}{D(t)},
\]
and therefore
\begin{equation}\label{eq:mh_individual_inequality}
    \frac{R(b)}{D(b)}
       \leq\int_{(b,v)}\frac{dK(t)}{D(t)}.
\end{equation}
The inequality also covers $H(b)=0$, since the first-price term in
\eqref{eq:mh_active_payoff} then vanishes. Furthermore, $D(b)>0$ for
almost every second-price-active action: if $H(b)=0$, then $K(b)=0$ and
$D(b)=mL(b)>0$ for $dL$-almost every such $b$.
Similarly, $D(t)>0$ for $dK$-almost every $t$.

Integrate \eqref{eq:mh_individual_inequality} with respect to the
subprobability measure of second-price-active actions. The distribution of their
first coordinates is $dL$. By Tonelli's theorem, the right-hand side is
\[
    \int \frac{\Pr[\mathcal A,\ B<t<V]}{D(t)}\,dK(t)
       =\int \frac{R(t)}{D(t)}\,dK(t).
\]
The equality uses the absence of atoms in $K$, so the countably many
possible atoms of $L$ do not affect this integral. Hence
\begin{equation}\label{eq:mh_nonpositive}
    I:=\int \frac{R(x)}{H(x)+mR(x)}\,dR(x)\leq0.
\end{equation}
Notice that no monotonicity, purity, or differentiability of the bidding
strategy has been used.

We next show that the left-hand side is strictly positive if second-price-active
actions have positive probability. For $\varepsilon>0$, define
\[
    \Phi_\varepsilon(h,r)
       :=\frac r m-
         \frac{h+\varepsilon}{m^2}
             \log\left(1+\frac{mr}{h+\varepsilon}\right),
         \qquad h,r\geq0.
\]
Its partial derivatives satisfy
\begin{align*}
    \partial_r\Phi_\varepsilon(h,r)
        &=\frac{r}{h+\varepsilon+mr},\\
    -\partial_h\Phi_\varepsilon(h,r)
        &=\frac1{m^2}
          \left[
            \log\left(1+\frac{mr}{h+\varepsilon}\right)
            -\frac{mr}{h+\varepsilon+mr}
          \right]\geq0,
\end{align*}
with strict inequality in the second line whenever $r>0$.
Also, $\Phi_\varepsilon$ is convex in $r$.

The path $x\mapsto(H(x),R(x))$ starts and ends with $R=0$, so the total
change of $\Phi_\varepsilon$ along the path is zero. Apply the
bounded-variation chain rule. On the continuous part of the path it is
the ordinary chain rule. At an upward jump of $R$, convexity gives
\[
    \Phi_\varepsilon(H,R)-\Phi_\varepsilon(H,R_-)
       \leq \partial_r\Phi_\varepsilon(H,R)(R-R_-).
\]
It follows that
\begin{equation}\label{eq:mh_positive_integral}
\begin{split}
    \int\frac{R}{H+\varepsilon+mR}\,dR
       &\geq -\int\partial_h\Phi_\varepsilon(H,R)\,dH\\
       &=\frac1{m^2}\int
          \left[
            \log\left(1+\frac{mR}{H+\varepsilon}\right)
            -\frac{mR}{H+\varepsilon+mR}
          \right]dH.
\end{split}
\end{equation}
Thus atoms in $L$ strengthen, rather than weaken, the required inequality.
This includes an atom at a first-price bid of zero.

To verify strictness, suppose $\Pr[\mathcal A]>0$. Since $B<V$ on
$\mathcal A$, there is a $c$ such that
$q:=\Pr[\mathcal A,\ B\leq c<V]>0$. Let $\nu$ be the subdistribution of
values on this event. It is atomless and has mass $q$. For $t>c$,
$R(t)\geq\nu((t,\overline v])$ and $dK\geq d\nu$. Therefore
\[
    \int R\,dH\geq\int R\,dK
       \geq\int \nu((t,\overline v])\,d\nu(t)
       =\frac{q^2}{2}>0.
\]
Consequently, the right-hand side of
\eqref{eq:mh_positive_integral} is strictly positive for every
$\varepsilon>0$. It is nondecreasing as $\varepsilon$ decreases.
On the left, the integrand is bounded by $1/m$ and $dR$ has finite total
variation. Dominated convergence as $\varepsilon\downarrow0$ therefore
gives $I>0$, contradicting \eqref{eq:mh_nonpositive}.
Thus $\Pr[\mathcal A]=0$.

It remains to identify the first coordinate. When $W\leq B$ almost
surely, the second-highest bid in the second-price auction cannot exceed
the highest first-price bid. The game therefore reduces to an ordinary first-price
auction, whose bid distribution is atomless by
Lemma~\ref{lem:mh_canonical}.
For completeness, let $C$ denote this bid distribution. The two
best-response inequalities for any two types imply that higher types
cannot use lower winning probabilities $C(b)^m$. Since $C(B)$ is uniform
on $[0,1]$ and is ordered with $V$, it equals $F(V)$ almost surely.
Thus equilibrium allocation probability is $F(v)^m$. Nonnegative bids imply
$0\leq U(v)\leq vF(v)^m$, so $U(v)\to0$ as $v\downarrow\underline v$.
The envelope formula therefore gives
\[
    U(v)=\int_{\underline v}^{v}F(t)^m\,dt.
\]
Each bid used by type $v$ must then satisfy
$(v-b)F(v)^m=U(v)$, which uniquely determines
$b=b^{\mathrm{FP}}(v)$. This proves purity of the first coordinate rather
than assuming it. The second-price coordinate can remain randomized below
this bid. Finally, against opponents using
$(b^{\mathrm{FP}}(V),0)$, a second-price bid cannot generate a profitable
allocation through the second-price auction, so the standard first-price
profile is an equilibrium of the full game.
\end{proof}

When agents are allowed to use Sybil identities, the equilibrium analysis
becomes more nuanced. The first-price-only equilibrium is still an
equilibrium even with Sybil strategies. However, uniqueness can generally
fail: an agent can submit the same bid through two identities in a
second-price auction. If those two identities are the highest bidders,
the agent pays her own bid, so the outcome effectively replicates a
first-price bid.

\subsection{Sealed and Open Bidding}
Next, we look into the question whether some form of open bidding could emerge in an ePBS world. There are two fundamental and related questions to be answered:
\begin{enumerate}
\item In the presence of a credible sealed bid first-price bidding channel, can an alternative auction format with open bidding survive?
\item If in the first-price bidding channel (generally we think here of the in-protocol channel), the seller could choose to leak received bids, would he do so?
\end{enumerate}
The first question is the English auction version of the question we studied in the last section. The analysis in this section applies to both single-homing and multi-homing environments; the participation distinction is orthogonal to the disclosure question. The answer in the symmetric case with equally fast bidders is that the English auction will unravel (or all bidders will bid at the ``last'' possible moment, essentially turning the auction into a sealed-bid first-price auction). In the asymmetric case, we will observe partial un-raveling: the slow bidder will self-select into the sealed-bidding channel, to avoid leaking information to bidders with a last look. The single- vs multi-homing distinction does not matter for these kinds of results. The conclusion from these results is that in the presence of a credible sealed-bid first-price bidding channel, the overall market is pushed in one way or another to a sealed-bid first-price auction architecture.

The answer to the second question will add important qualifications to the conclusion drawn from the answer to the first question: whether or not the seller has an incentive to keep the first-price bidding channel sealed depends on how many bidders have a last look advantage. In the symmetric case, as well as in the case of one fast bidder with last look, the seller has no incentive to leak bids. In the case of at least two fast bidders, it generally has an incentive to leak bids. 
Practically, we could imagine different ways of seller bid leakage in ePBS, e.g. the proposer could forward bids to a relay that integrates them into its bid stream.

However, if the seller can credibly commit to not leak information (e.g. through reputation, TEE commitments or delegation to an intermediary such as a possible ``first-price sealed-bid relay'') it is optimal for him to do so.

We formalize the unraveling of open-bidding relays in Propositions~\ref{nolastlook} and~\ref{cor:open_relay_no_leakage}, and the seller's disclosure incentives (without and with commitment) in Propositions~\ref{1leak} and~\ref{prop:commitment_show}  and Theorem~\ref{2leak}.

\subsubsection{No Last Look}
\label{sec:english_without_last_look}

We begin with a deliberately relay-friendly idealization. The relay is allowed to run a fully converged English auction among its participants, but the openness of the auction is purely internal: after the relay selects its representative, the relay representative and the native-channel builders submit proposer-facing bids without observing each other's final bids. Thus the model removes any cross-channel last-look advantage. The purpose of this benchmark is to ask whether an English relay can be stable even in this favorable case. Proposition~\ref{nolastlook} shows that it cannot: in large markets, all builders joining the English relay is not an equilibrium.

Fix a first-stage participation profile. Let $S$ be the set of builders who choose the relay, and let $A=[n]\setminus S$ be the set of builders who choose the native first-price channel. Write $s=|S|$ and $m=|A|$.

If $s\geq 2$, the relay runs an internal English knockout auction. Let $Y_S$ be the highest value among relay participants and let $Z_S$ be the second-highest value. The internal auction selects the builder with value $Y_S$ and sets the internal English-auction price equal to $Z_S$. The relay winner then chooses a proposer-facing bid $b_S$. Since bidding below $Z_S$ would not be consistent with the internal English outcome, while bidding above $Y_S$ is weakly dominated, we restrict attention to bids satisfying
\begin{equation*}
    Z_S\leq b_S\leq Y_S.
\end{equation*}
If $s=1$, the relay has no internal competition and therefore generates no English-auction information; the sole relay participant simply chooses a proposer-facing bid subject to the usual no-overbidding constraint.

Each builder $j\in A$ submits a native first-price bid $b_j$ to the proposer. The proposer compares the relay bid $b_S$, when present, with the native bids $(b_j)_{j\in A}$ and accepts the highest bid. The key modeling restriction is that the relay representative chooses $b_S$ without observing the native bids, and native builders bid without observing the internal relay state. Hence the English auction reveals information only within the relay and creates no last-look opportunity across channels.

To test stability of the all-English profile, let $U_E^n(F)$ be a builder's ex-ante payoff when all $n$ builders choose the relay, and let $U_{\mathrm{out}}^n(F)$ be the ex-ante payoff of a single builder who deviates to the native first-price channel while the other $n-1$ builders remain in the relay. The all-English profile can be a first-stage equilibrium only if
\begin{equation*}
    U_E^n(F)\geq U_{\mathrm{out}}^n(F).
\end{equation*}

The all-first-price profile is a Bayes--Nash equilibrium of the two-stage game: a unilateral deviation to the relay creates no internal relay competition and hence no informational advantage, so the deviation is payoff-equivalent to remaining in the native first-price channel. The all-English profile is less robust. For compactly supported distributions with non-flat upper tails, full participation in the English relay is unstable once the number of builders is sufficiently large.

\begin{proposition}[Asymptotic instability of all-English]\label{nolastlook}
There exists $N(F)<\infty$ such that, for every $n\geq N(F)$, the all-English
profile is not a first-stage equilibrium.
\end{proposition}

Therefore, even in the idealized no-last-look benchmark, the relay cannot rely on single-homing alone to make an English format stable against a native sealed first-price alternative. The next question is what happens when the relay or proposer can create a genuine last-look opportunity by disclosing early bids.
\subsubsection{Last Look}
\label{sec:last_look_no_commitment}

We now move to the more realistic case in which a subset of builders is latency-advantaged. These builders can submit a final response after observing information released close to the auction deadline; in this sense, they have the technological ability to obtain \emph{last look}. Last look is not automatic, however. It matters only if the seller, the relay, or the auction format reveals relevant bid information
before the deadline.

We first consider the case in which the seller, interpreted here as the block proposer, cannot credibly commit ex-ante to a disclosure strategy. As highlighted in Footnote~\ref{fn:no_commitment}, this does not exclude the possibility that the seller chooses their strategy ex-ante, but it forces them to choose a (Nash) equilibrium strategy. Let $L\subseteq[n]$ denote the set of builders who are technologically able to respond to late information before the auction closes. A builder in $L$ does not automatically have last look; it obtains last look only if relevant information is revealed before the deadline. In the model below, we write $S\subseteq L$ for the builders who may receive such disclosures.

The game has three stages. First, the builders in $[n]\setminus S$ submit bids to the seller. Second, after observing these bids, the seller chooses a verifiable message $\sigma_i$ to send to each bidder $i\in S$, where $\sigma_i$ consists of a collection of bid pairs $\{(j,b_j)\}$ corresponding to an arbitrary subset of $[n]\setminus S$. Third, the bidders in $S$ bid simultaneously. The block is allocated to the highest bidder under a pay-as-bid rule, with ties broken lexicographically in favor of bidders in $S$.\footnote{We could formulate a more general version of the game by adding a fourth stage in which the seller chooses which bid to accept. However, in any sequential equilibrium, the seller would choose the highest bid. Hence, without loss of generality, we restrict attention to the three-stage game.}

Figure~\ref{fig:last_look_game} summarizes the three-stage game.

We say that an auction is \emph{$k$-leakage-resistant} if, for every set $S$ with $|S|=k$, in every sequential Bayes--Nash equilibrium the seller reveals no information in the second stage. We say that it is \emph{$k$-leaking} if for every set $S$ with $|S|=k$, and every sequential Bayes--Nash equilibrium, the seller leaks information with probability one.

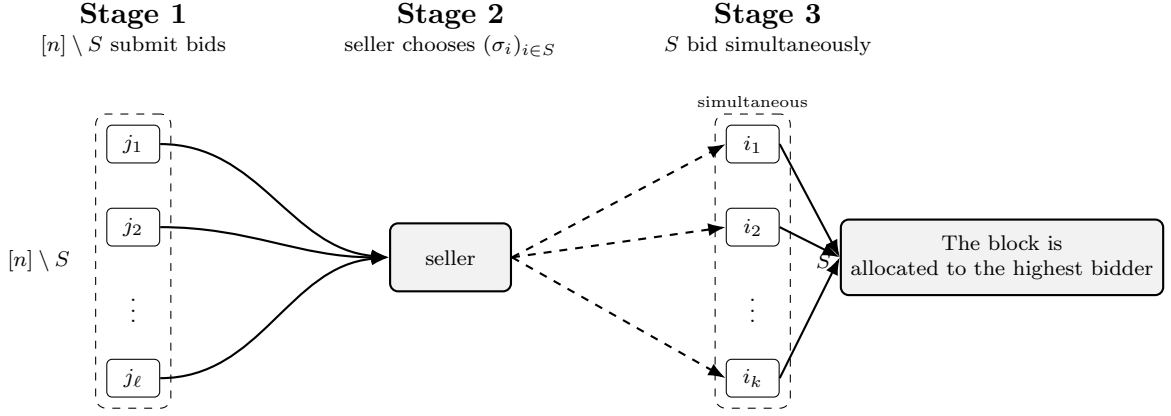
\begin{figure}
    \centering
\begin{tikzpicture}[
    scale=1, transform shape,
    xshift=-4cm,
    >=Latex,
    bidder/.style={
      draw,
      rounded corners=2pt,
      minimum width=7mm,
      minimum height=5mm,
      align=center,
      font=\scriptsize,
      fill=white
    },
    block/.style={
      draw,
      thick,
      rounded corners=3pt,
      minimum width=16mm,
      minimum height=9mm,
      align=center,
      font=\scriptsize,
      fill=gray!10
    },
    group/.style={
      draw,
      dashed,
      rounded corners=4pt,
      inner sep=4pt
    },
    bid/.style={->, thick},
    info/.style={->, thick, dashed},
    every node/.style={font=\scriptsize}
]

\node[font=\bfseries] at (0.0,3.0)  {Stage 1};
\node[font=\bfseries] at (4.2,3.0)  {Stage 2};
\node[font=\bfseries] at (8.4,3.0)  {Stage 3};

\node at (0.0,2.6)  {$[n]\setminus S$ submit bids};
\node at (4.2,2.6)  {seller chooses $(\sigma_i)_{i\in S}$};
\node at (8.4,2.6) {$S$ bid simultaneously};

\node[bidder] (j1) at (0,1.3)   {$j_1$};
\node[bidder] (j2) at (0,0.2)   {$j_2$};
\node          (jd) at (0,-0.8) {$\vdots$};
\node[bidder] (jm) at (0,-1.8)  {$j_\ell$};

\node[group, fit=(j1)(j2)(jd)(jm),
      label={[xshift=-2mm]left:$[n]\setminus S$}] {};

\node[block] (seller) at (4.2,-0.2) {seller};

\draw[bid] (j1.east) to[out=0,in=175] (seller.west);
\draw[bid] (j2.east) to[out=0,in=180] (seller.west);
\draw[bid] (jm.east) to[out=0,in=185] (seller.west);

\node[bidder] (i1) at (8.2,1.3)   {$i_1$};
\node[bidder] (i2) at (8.2,0.2)   {$i_2$};
\node          (id2) at (8.2,-0.8) {$\vdots$};
\node[bidder] (ik) at (8.2,-1.8)  {$i_k$};

\node[group, fit=(i1)(i2)(id2)(ik),
      label={[xshift=2mm]right:$S$}] {};
\node[font=\tiny, above=1mm of i1] {simultaneous};

\draw[info] (seller.east) -- (i1.west);
\draw[info] (seller.east) -- (i2.west);
\draw[info] (seller.east) -- (ik.west);

\node[block, minimum width=26mm, minimum height=10mm]
(outcome) at (11.5,-0.2)
{\scriptsize The block is\\
 \scriptsize allocated to the highest bidder};

\draw[bid] (i1.east) -- (outcome.west);
\draw[bid] (i2.east) -- (outcome.west);
\draw[bid] (ik.east) -- (outcome.west);
\end{tikzpicture}
\caption{Last look vs First-price auction game representation.}
\label{fig:last_look_game}
\end{figure}

\begin{proposition}\label{1leak}
The first-price auction is $1$-leakage-resistant.
\end{proposition}
\begin{proof}
Let $S=\{i\}$ and let $r$ denote the highest bid submitted in stage 1. If the seller reveals $r$, then bidder $i$ can condition on it and, whenever $v_i\geq r$, win by bidding exactly $r$, since ties are broken in favor of $S$. Thus leakage yields revenue exactly $r$. If instead the seller reveals nothing, bidder $i$ must choose his bid without observing the realized value of $r$. Hence, with positive probability, his equilibrium bid is strictly above $r$, and on that event the seller receives strictly more than $r$. Therefore the seller's expected revenue is strictly higher under no leakage, so in any sequential equilibrium the seller does not leak.
\end{proof}

With a unique latency-advantaged builder, the seller's incentive is compatible with no disclosure: the first-price allocation is preserved, builders are ex-ante symmetric, and the mechanism weakens incentives to compete purely on latency rather than on order-flow acquisition or block construction. However, with multiple latency-advantaged builders this is no longer true.

\begin{theorem}\label{2leak}
The first-price auction is \emph{$k$-leaking} for all $2 \leq k \leq n-1$.
\end{theorem}
\begin{proof}[Proof of the theorem] Let us fix a sequential Bayes--Nash equilibrium and write $s=|S|$. Let $R$ be the random variable
\begin{equation*}
    R:=\max\{b_j:j\in [n]\setminus S\},
\end{equation*}
that is, the highest bid among the bidders in the first stage. Let $M$ denote the seller's message in the second stage, with $M=\emptyset$ denoting the event that the seller reveals no information at that stage. Define
\begin{equation*}
    \overline{r}=\operatorname*{ess\,sup}(R\mid M=\emptyset).
\end{equation*}
Observe that first-stage bids are bounded away from $\overline v$. Fix $a\in (\underline v, \overline{v})$, let $q=F(a)^{n-1}$, and choose $\eta>0$ such that $(\overline{v}-\eta-a)q>\eta$. A bidder with value $v\geq \overline v - \eta$ can bid $a$ and win whenever all other values are below $a$, obtaining a payoff at least $(\overline v-\eta -a)q>\eta$. On the other hand, any bid $x>\overline v - \eta$ gives $v-x<\eta$. Hence no high type bids above $\overline v-\eta$, while lower types do not do so by no-overbidding. Therefore $R\leq \overline{v}-\eta$ almost surely, and so $\overline{r}<\overline{v}$.

\begin{lemma}\label{lemma:reduction}
Fix the continuation game among the bidders in $S$ after stage 2. Consider the two cases in which the seller either reveals the realized highest first-stage bid $r$ to all bidders in $S$, or reveals nothing. In both cases, every equilibrium outcome is outcome-equivalent to a symmetric Bayes--Nash equilibrium. Moreover, every symmetric mixed Bayes--Nash equilibrium is outcome-equivalent to a pure symmetric Bayes--Nash equilibrium.
\end{lemma}

\begin{proof}
Conditional on either public message, the bidders in $S$ are ex-ante identical, their values are still i.i.d. according to $F$, and the continuation game is a continuous first-price auction game with a stochastic or static reserve price. Hence the symmetrization argument in \cite{chawla2013auctions} applies, so we may restrict attention to symmetric equilibria. The same result also implies that any symmetric mixed equilibrium is outcome-equivalent to a pure symmetric equilibrium. Since the seller's continuation payoff depends only on the induced distribution of bids and payments, this replacement is without loss of generality.
\end{proof}

By Lemma~\ref{lemma:reduction}, we may assume without loss of generality that the continuation bidding strategies are symmetric and pure.

We compare the continuation game after no disclosure with the continuation game in which the seller reveals the realized highest bid to all bidders in $S$. Let the corresponding symmetric pure bidding strategies be $b_{\emptyset}$ and $b_r$.

When the seller reveals $M=r$ to all bidders in $S$, the continuation game is the standard first-price auction with reserve $r$. Its symmetric equilibrium is unique, and each player has bidding strategy
\begin{equation*}
  b_r(v)= \begin{cases}
  v-\dfrac{\int_r^v F(t)^{s-1}dt}{F(v)^{s-1}},&v\geq r,\\
  0,&v<r.
  \end{cases}
\end{equation*}

Now consider the continuation game after no disclosure. Let
\begin{equation*}
    H(x)=\Pr[R\leq x \mid M=\emptyset]
\end{equation*}
be the cdf of $R$ conditional on $M=\emptyset$ and $b_\emptyset$ the bid function of the symmetric Bayes-Nash equilibrium with $M=\emptyset$. Define
\begin{equation*}
    v^\star=\inf\{v:b_{\emptyset}(v)\geq \overline{r}\}.
\end{equation*}
A first observation is that $v^\star>\overline{r}$. Indeed, in a first-price auction bidding above one's value is weakly dominated, so $b_{\emptyset}(v)\leq v$ for every $v$, which already implies $v^\star\geq \overline{r}$. If $v^\star=\overline{r}$, then by continuity we would have $b_{\emptyset}(\overline{r})=\overline{r}$. But then a bidder of type $\overline{r}$ obtains payoff $0$ from bidding $\overline{r}$, while by shading slightly below $\overline{r}$ she still beats $R$ with positive probability under $M=\emptyset$ and wins against the other bidders in $S$ with positive probability, yielding strictly positive payoff. This contradicts optimality, so $v^\star>\overline{r}$.

\proofsubparagraph*{Case 1: $v^\star<+\infty$.} For all $v\geq v^\star$, equilibrium bids satisfy $b_{\emptyset}(v)\geq \overline{r}$, so such bids beat the first-stage maximum whenever $M=\emptyset$. Hence, if a bidder of type $v$ deviates to the bid of a type $z\geq v^\star$, her payoff is
\begin{equation*}
    (v-b_{\emptyset}(z))F(z)^{s-1}.
\end{equation*}
Optimality at $z=v$ gives
\begin{equation*}
    \frac{d}{dz}\Big((v-b_{\emptyset}(z))F(z)^{s-1}\Big)\bigg|_{z=v}=0,
\end{equation*}
that is,
\begin{equation*}
    -b_{\emptyset}'(v)F(v)^{s-1}+(v-b_{\emptyset}(v))(s-1)F(v)^{s-2}f(v)=0.
\end{equation*}
Equivalently,
\begin{equation*}
    \frac{d}{dv}\Big((v-b_{\emptyset}(v))F(v)^{s-1}\Big)=F(v)^{s-1}.
\end{equation*}
Using the boundary condition $b_{\emptyset}(v^\star)=\overline{r}$, we obtain
\begin{equation*}
    b_{\emptyset}(v)=v-\frac{(v^\star-\overline{r})F(v^\star)^{s-1}+\int_{v^\star}^v F(t)^{s-1}dt}{F(v)^{s-1}}.
\end{equation*}

We now compare the two bid functions. For every $v\geq v^\star$,
\begin{equation*}
    b_{\overline{r}}(v)-b_{\emptyset}(v)=\frac{(v^\star-\overline{r})F(v^\star)^{s-1}-\int_{\overline{r}}^{v^\star} F(t)^{s-1}dt}{F(v)^{s-1}}.
\end{equation*}
Since $F$ is strictly increasing, we have $F(t)^{s-1}<F(v^\star)^{s-1}$ for all $t\in[\overline{r},v^\star)$, and therefore
\begin{equation*}
    \int_{\overline{r}}^{v^\star} F(t)^{s-1}dt<(v^\star-\overline{r})F(v^\star)^{s-1}.
\end{equation*}
Hence
\begin{equation*}
    b_{\overline{r}}(v)>b_{\emptyset}(v)\qquad \text{for all } v\geq v^\star.
\end{equation*}
Also, if $v\in(\overline{r},v^\star)$, then by definition $b_{\emptyset}(v)<\overline{r}$, whereas
\begin{equation*}
    b_{\overline{r}}(v)=v-\frac{\int_{\overline{r}}^v F(t)^{s-1}dt}{F(v)^{s-1}}>\overline{r}.
\end{equation*}
Thus
\begin{equation*}
    b_{\overline{r}}(v)>b_{\emptyset}(v)\qquad \text{for all } v>\overline{r}.
\end{equation*}

\proofsubparagraph*{Case 2: $v^\star=+\infty$.} In this case, for all $v\in[\underline{v},\overline{v}]$, $b_{\emptyset}(v)<\overline{r}$. However, for $v\in[\overline{r},\overline{v}]$, $b_{\overline{r}}(v)\geq \overline{r}$, and therefore $b_{\overline{r}}(v)>b_{\emptyset}(v)$.

Let $v_i$ denote the valuation of bidder $i\in S$. The seller's continuation revenue after revealing $r$ is
\begin{equation*}
    \Pi_{\mathrm{rev}}(r)=\mathbb{E}\Big[\max\{r,\max_{i\in S} b_r(v_i)\}\Big],
\end{equation*}
whereas after no disclosure it is
\begin{equation*}
    \Pi_{\emptyset}(r)=\mathbb{E}\Big[\max\{r,\max_{i\in S} b_{\emptyset}(v_i)\}\Big].
\end{equation*}
From the pointwise comparison above, for every valuation profile $(v_i)_{i\in S}$,
\begin{equation*}
    \max\{\overline{r},\max_{i\in S} b_{\overline{r}}(v_i)\}\geq \max\{\overline{r},\max_{i\in S} b_{\emptyset}(v_i)\},
\end{equation*}
and the inequality is strict whenever $\max_{i\in S}v_i>\overline{r}$. Since $\overline{r}<\overline{v}$,
\begin{equation*}
    \Pr\Big[\max_{i\in S}v_i>\overline{r}\Big]=1-F(\overline{r})^s>0.
\end{equation*}
Therefore
\begin{equation*}
    \Pi_{\mathrm{rev}}(\overline{r})>\Pi_{\emptyset}(\overline{r}).
\end{equation*}

The seller's continuation revenue is continuous in $r$ in both cases. Hence there exists $\delta>0$ such that
\begin{equation*}
    \Pi_{\mathrm{rev}}(r)>\Pi_{\emptyset}(r)\qquad \text{for all } r\in[\overline{r}-\delta,\overline{r}].
\end{equation*}
By definition of $\overline{r}$ as the essential supremum of $R$ conditional on $M=\emptyset$,
\begin{equation*}
    \Pr[R\in[\overline{r}-\delta,\overline{r}]\mid M=\emptyset]>0.
\end{equation*}
Hence, on a positive-probability set of histories at which the equilibrium prescribes no disclosure, the seller would strictly prefer to reveal the realized highest bid. This contradicts sequential optimality. Therefore the seller cannot choose $M=\emptyset$ with positive probability.
\end{proof}

The previous theorem shows that the allocation induced by the three-stage game is not equivalent to the allocation in the standard first-price auction. Since the first-price auction is welfare-optimal with symmetric i.i.d. bidders, the three-stage game is not welfare-optimal. Moreover, when value distributions are regular, the first-price auction is revenue-optimal within the class of auctions without a reserve price (see Proposition \ref{prop:commitment_show}). Thus allowing the seller to leak information cannot improve on the no-disclosure benchmark; the theorem implies that revenue is strictly lower. Also, information disclosure can have an adverse effect on competition
in the block-building market by increasing market concentration. As we show
in the following proposition, disclosure raises the payoff of agents with an
informational advantage, thereby strengthening their position relative to
less-informed agents.

\begin{proposition}[Last-look rents]\label{lastlookrents}
Suppose the seller reveals the highest first-stage bid to every bidder in $S$. Then, in any symmetric equilibrium, every bidder in $S$ obtains a strictly larger ex-ante utility and probability of winning than every bidder in $[n]\setminus S$.
\end{proposition}
\begin{proof}
As in Lemma~\ref{lemma:reduction}, the symmetrization argument in \cite{chawla2013auctions} applies to the bidders in $A:=[n]\setminus S$. Let $m:=|A|$ and $s:=|S|$, so that $m+s=n$. Therefore, we can reduce the analysis to the case in which bidders in $A$ use a monotone bidding strategy $b_A$ and bidders in $S$ use a monotone bidding strategy $b_S$. Let $X_A(v)$ and $X_S(v)$ denote the interim winning probabilities of a bidder with value $v$ in groups $A$ and $S$, respectively.

Almost surely, every player bids strictly below her valuation. In particular, $b_A(v)<v$. A bidder in $A$ with value $v$ wins only if the other bidders in $A$ have values below $v$ and the bidders in $S$ have values below $b_A(v)$. Hence
\begin{equation*}
    X_A(v)=F(v)^{m-1}F(b_A(v))^s\leq F(v)^{m-1}F(v)^s=F(v)^{n-1},
\end{equation*}
with strict inequality on a positive-probability set.

By contrast, a bidder in $S$ with value $v$ wins whenever the other bidders in $S$ have values below $v$ and the highest first-stage bid is at most $v$. Therefore
\begin{equation*}
    X_S(v)=F(v)^{s-1}\Pr\Big[\max_{j\in A} b_A(v_j)\leq v\Big].
\end{equation*}
Since $b_A(t)\leq t$, the event that all bidders in $A$ have values below $v$ implies $\max_{j\in A}b_A(v_j)\leq v$. Hence
\begin{equation*}
    X_S(v)\geq F(v)^{s-1}F(v)^m=F(v)^{n-1}.
\end{equation*}
Thus $X_S(v)>X_A(v)$ on a positive-measure set. By the envelope theorem,
\begin{equation*}
    u_S(v)=\int_{\underline v}^{v}X_S(t)\,dt
    \geq
    \int_{\underline v}^{v}X_A(t)\,dt
    =u_A(v),
\end{equation*}
with strict inequality on a positive-measure set. Taking expectations over $v\sim F$ gives the result.
\end{proof}

\subsubsection{Last-Look with Seller Commitment}
\label{sec:last_look_commitment}

The previous section studied the case in which the seller could not commit to a disclosure policy. In that environment, disclosure is chosen history by history after early bids have arrived. We showed that this creates an ex-post incentive to leak information when there are at least two latency-advantaged bidders: after observing the early bids, the seller can raise continuation revenue by revealing information to the bidders who are still able to respond.

We now ask a different question. Suppose instead that the seller can commit ex-ante to a particular auction rule, or equivalently can choose among relays that commit to different rules (and make it credible that they will ignore or not-leak direct bids sent to them). One relay may run a sealed first-price auction, another may run an open English-style auction, and another may commit to giving last look to a specified set of builders. In this sense, the seller's decision is over which auction format they prefer to pick ex-ante and commit to.

This distinction matters. Ex-post, after early bids have been observed, leaking information can be profitable. Ex-ante, however, the proposer must take into account that last look changes bidders' incentives and shifts surplus toward the latency-advantaged builders. The result below shows that, when values are i.i.d. from a regular distribution, this ex-ante effect dominates: among committed last-look rules, the proposer weakly prefers the sealed first-price auction.

We assume that the seller can choose among the following auction formats.
For a set $S\subseteq[n]$, define the auction that reveals bids from $[n]\setminus S$ to bidders in $S$ as follows:
\begin{enumerate}
    \item The set of builders $A := [n] \setminus S$ submits bids simultaneously.
    \item The bids submitted by $A$ are revealed to the builders in $S$.
    \item The builders in $S$ then submit bids simultaneously.
    \item The right to decide the content of the next block is allocated to the highest bidder under a pay-as-you-bid payment rule.
\end{enumerate}
Subsequently, we use the short-hand notation $\texttt{Show}(S)$ to denote the auction that reveals bids from $[n]\setminus S$ to bidders in $S$. 
Observe that the cases $S=\emptyset$ and $S=[n]$ are equivalent to the first-price auction. For the set of latency-advantaged parties $L\subseteq [n]$, one can view the case $S=L$ as the current status quo, i.e. an open first-price auction with last-look advantage for a set of parties. 

First, we assume that the block proposer can \textit{commit} to a specific auction format; that is, once they announce to whom they reveal bids, they cannot deviate from it during its execution. 

\begin{definition}
A distribution $F$ is \emph{regular} if its virtual valuation function
\begin{equation*}
\phi(v) = v - \frac{1-F(v)}{f(v)}
\end{equation*}
is non-decreasing in $v$, where $f$ is the density of $F$.
\end{definition}
\begin{proposition}\label{prop:commitment_show} 
If $F$ is regular, then for any proper set $S\subseteq[n]$ and any
Bayes--Nash equilibrium of $\texttt{Show}\,(S)$,
\begin{equation*}
    \mathrm{Rev}(\mathrm{FP})\geq \mathrm{Rev}(\texttt{Show}\,(S)).
\end{equation*}
In other words, under commitment the block proposer weakly prefers the sealed
first-price auction.
\end{proposition}

Suppose, finally, that the proposer commits not to leak bids submitted through
the sealed first-price channel, but still accepts bids coming from an
open-bid relay. In this case, a high-latency builder has no reason to send her
bid through the open relay. If she sends a bid $b$ through the relay,
low-latency builders may observe it and still have time to react. If she sends
the same bid $b$ through the sealed channel, the proposer compares it in the
same way, but the bid is not revealed before the deadline. Hence, whenever the
builder has a chance of winning, submitting through the relay can only reduce
that chance of winning.

\begin{proposition}[Open relays under no leakage]
\label{cor:open_relay_no_leakage}
Suppose the proposer commits not to leak bids submitted through the sealed
first-price channel, while still accepting bids from an open-bid relay. Then
high-latency builders bid through the sealed channel. Low-latency builders may
still use or observe the open relay, but the relay no longer gives them a
last-look advantage over sealed bids. The outcome is equivalent to the
standard sealed first-price auction.
\end{proposition}

\section{Conclusion}\label{sec:conclusion}
We have analyzed competing auctions in intermediated markets where a single seller can route a single sale through multiple parallel mechanisms, with proposer-builder separation as the leading application. We found that second-price intermediary auctions unravel fully against a sealed first-price principal auction (Theorem~\ref{SPUnraveling} under single-homing, Theorem~\ref{Essential} under multi-plexing), open-bidding intermediaries unravel partially - collapsing into first-price under symmetric latency, sorting fast bidders to the intermediary under asymmetric latency (Propositions~\ref{nolastlook} and~\ref{cor:open_relay_no_leakage}) - and Sybil bidding can restore the second-price intermediary by allowing bidders to convert it into a first-price-equivalent mechanism. On the disclosure side, a first-price auction is leakage-resistant against a single privileged bidder (Proposition~\ref{1leak}) but leaks in equilibrium against two or more (Theorem~\ref{2leak}), and the optimal commitment under regularity is to share no information at all (Proposition~\ref{prop:commitment_show}).

The institutional reading of these results is that ePBS will likely reshape the relay market: An in-protocol sealed-bid first-price channel limits the design space for relays going forward. Relays running second-price or open-bidding mechanisms cannot sustain participation against the sealed first-price alternative. However, this prediction relies on the in-protocol channel being credibly sealed; the assumption of a non-leaking in-protocol bidding channel is itself fragile: the disclosure analysis identifies that a proposer receiving direct bids has an incentive to leak bids in several scenarios. Thus, proposers forwarding bids to a relay's open-bid feed is a possible outcome. On the other hand, there is value for proposers from a credible commitment not to
leak bids. This could be created through reputation mechanisms or solutions such as TEEs, in the case of stakers large enough to provide individual commitments (Lido and similar institutional stakers), or through intermediation in the form of ``sealed-bid'' relays. Thus, there are several factors pushing relays to effectively move to sealed-bid first-price auction formats post ePBS. Still, there is a potential role for them as credible intermediaries and for services such as fast bid and block propagation as well as non-collateralized bidding. There is also, in principle, the narrower path to maintain an open-bid relaying service to which proposers forward bids. However, the viability of this latter path depends crucially on the nature of competition in the builder market.

 Several limitations of our analysis are worth flagging. We do not model collusion among builders, which is a standard concern for sealed-bid auctions in classical settings but plausibly less binding in blockchain environments where on-chain transparency and programmable coordination weaken the open-vs-sealed distinction as a collusion-suppression tool. We do not model the searcher-builder pipeline upstream of the auction; our values are private at the builder level, which is the relevant abstraction for the auction itself but abstracts from common-value dynamics among searchers. We treat each slot in isolation rather than as part of a repeated game; reputational mechanisms that sustain relay commitment are mentioned but not formally modeled. The analysis also takes the set of available mechanisms and intermediaries as given; the entry and design problem for relays - what mechanism a new relay should propose, given the equilibrium predictions above - is a natural follow-up question.

\bibliographystyle{RM}
\bibliography{bibliography}

@article{McAfee1993,
  title={Mechanism design by competing sellers},
  author={McAfee, R Preston},
  journal={Econometrica: Journal of the econometric society},
  pages={1281--1312},
  year={1993},
  publisher={JSTOR}
}

@article{PetersSeverinov1997,
  title={Competition among sellers who offer auctions instead of prices},
  author={Peters, Michael and Severinov, Sergei},
  journal={Journal of Economic Theory},
  volume={75},
  number={1},
  pages={141--179},
  year={1997},
  publisher={Elsevier}
}

@article{RochetTirole2006,
  title={Two-sided markets: a progress report},
  author={Rochet, Jean-Charles and Tirole, Jean},
  journal={The RAND journal of economics},
  volume={37},
  number={3},
  pages={645--667},
  year={2006},
  publisher={Wiley Online Library}
}

@article{Armstrong2006,
  title={Competition in two-sided markets},
  author={Armstrong, Mark},
  journal={The RAND journal of economics},
  volume={37},
  number={3},
  pages={668--691},
  year={2006},
  publisher={Wiley Online Library}
}

@article{CaillaudJullien2003,
  title={Chicken \& egg: Competition among intermediation service providers},
  author={Caillaud, Bernard and Jullien, Bruno},
  journal={RAND journal of Economics},
  pages={309--328},
  year={2003},
  publisher={JSTOR}
}

@book{FoucaultPaganoRoell2023,
  title={Market liquidity: theory, evidence, and policy},
  author={Foucault, Thierry and Pagano, Marco and R{\"o}ell, Ailsa},
  year={2023},
  publisher={Oxford University Press}
}

@article{MilgromWeber1982,
  title={The value of information in a sealed-bid auction},
  author={Milgrom, Paul and Weber, Robert J},
  journal={Journal of Mathematical Economics},
  volume={10},
  number={1},
  pages={105--114},
  year={1982},
  publisher={Elsevier}
}

@article{Oomen2017,
  title={Last look},
  author={Oomen, Roel},
  journal={Quantitative Finance},
  volume={17},
  number={7},
  pages={1057--1070},
  year={2017},
  publisher={Taylor \& Francis}
}

@article{OHaraYe2011,
  title={Is market fragmentation harming market quality?},
  author={O'Hara, Maureen and Ye, Mao},
  journal={Journal of Financial Economics},
  volume={100},
  number={3},
  pages={459--474},
  year={2011},
  publisher={Elsevier}
}

@article{BergemannPesendorfer2007,
  title={Information structures in optimal auctions},
  author={Bergemann, Dirk and Pesendorfer, Martin},
  journal={Journal of economic theory},
  volume={137},
  number={1},
  pages={580--609},
  year={2007},
  publisher={Elsevier}
}

@article{BurguetSakovics1999,
  title={Imperfect competition in auction designs},
  author={Burguet, Roberto and S{\'a}kovics, J{\'o}zsef},
  journal={International Economic Review},
  volume={40},
  number={1},
  pages={231--247},
  year={1999},
  publisher={Wiley Online Library}
}

@article{ellison2004competing,
  title={Competing auctions},
  author={Ellison, Glenn and Fudenberg, Drew and M{\"o}bius, Markus},
  journal={Journal of the European Economic Association},
  volume={2},
  number={1},
  pages={30--66},
  year={2004},
  publisher={Oxford University Press}
}

@inproceedings{chawla2013auctions,
  title={Auctions with unique equilibria},
  author={Chawla, Shuchi and Hartline, Jason D},
  booktitle={Proceedings of the fourteenth ACM conference on Electronic commerce},
  pages={181--196},
  year={2013}
}

@inproceedings{pai2024structural,
  title={Structural advantages for integrated builders in mev-boost},
  author={Pai, Mallesh and Resnick, Max},
  booktitle={International Conference on Financial Cryptography and Data Security},
  pages={128--132},
  year={2024},
  organization={Springer}
}

@misc{eip7732,
  author       = {Francesco D'Amato and Nico Flaig and Barnab{\'e} Monnot and Michael Neuder and Potuz and Justin Traglia and Terence Tsao},
  title        = {{EIP-7732: Enshrined Proposer-Builder Separation}},
  year         = {2024},
  howpublished = {\url{https://eips.ethereum.org/EIPS/eip-7732}},
  note         = {Draft EIP, accessed April 2026}
}

@article{mazorra2025free,
  title={The Free Option Problem of ePBS},
  author={Mazorra, Bruno and {\"O}z, Burak and Schlegel, Christoph and Wu, Fei},
  journal={arXiv preprint arXiv:2509.24849},
  year={2025}
}
\appendix
\section{Proofs}
\label{app:proofs}

\subsection{Proof of Lemma~\ref{lem:fp_bid_monotone_single_homing}}
As it is optimal for bidder $v$ to bid $b(v)$, we have
\begin{equation*}
    (v-b(v))\Pr[\text{winning with }b(v)]\geq (v-b(w))\Pr[\text{winning with }b(w)].
\end{equation*}
Similarly, as it is optimal for bidder $w$ to bid $b(w)$, we have
\begin{equation*}
    (w-b(w))\Pr[\text{winning with }b(w)]\geq (w-b(v))\Pr[\text{winning with }b(v)].
\end{equation*}
Summing these inequalities and rearranging, we get
\begin{equation*}
    (v-w)(\Pr[\text{winning with }b(v)]-\Pr[\text{winning with }b(w)])\geq0.
\end{equation*}
Therefore $\Pr[\text{winning with }b(v)]\geq\Pr[\text{winning with }b(w)]$.
Now suppose towards contradiction that $b(v)<b(w)$. The winning probability is weakly increasing in the bid, therefore
\begin{equation*}
  \Pr[\text{winning with }b(v)]\leq \Pr[\text{winning with }b(w)].
\end{equation*}
Hence, $\Pr[\text{winning with }b(v)]=\Pr[\text{winning with }b(w)]$. Since $v,w>\underline{v}$, this quantity is strictly positive. The optimality inequality for type $w$ then implies $b(v)\geq b(w)$, a contradiction.

\subsection{Proof of Proposition~\ref{nolastlook}}
\begin{proof}
First compute the payoff from staying in the all-English profile. If all
builders use the relay, the English auction selects the highest value and the
winner pays the second-highest value. Hence
\begin{equation*}
U_E^n(F)
=
\int_{\underline v}^{\bar v}
(1-F(v))F(v)^{n-1}\,dv .
\end{equation*}
Changing variables $u=F(v)$ and then $u=1-x/n$ gives
\begin{equation*}
U_E^n(F)
=
\frac{1}{n^2}
\int_0^n
x\left(1-\frac{x}{n}\right)^{n-1}
\frac{1}{f(F^{-1}(1-x/n))}\,dx .
\end{equation*}
Since $f(F^{-1}(1-x/n))\to f(\overline v)$ and
$\left(1-x/n\right)^{n-1}\to e^{-x}$ for each fixed $x$, a standard
upper-tail splitting argument gives
\begin{equation*}
n^2f(\overline v)\,U_E^n(F)\to \int_0^\infty xe^{-x}\,dx=1 .
\end{equation*}

Now consider a single native-channel deviator against a relay of size $n-1$.
Let $X$ be the highest relay value and $Y$ the second-highest relay value.
The relay representative submits
\begin{equation*}
b_R(X,Y)=\max\{Y,\alpha_n(X)\},
\end{equation*}
where $\alpha_n(X)$ is the representative's voluntary bid.

We first note that the voluntary bid is bounded away from the top of the
support. Pick any $a<\bar v$ with $F(a)>0$, and set
\begin{equation*}
\eta=\frac{(\bar v-a)F(a)}{2(1+F(a))}.
\end{equation*}
Then $\eta>0$ and $(\bar v-\eta-a)F(a)>\eta$. We claim that
$\alpha_n(X)\leq \bar v-\eta$ for every feasible $X$. If
$X\leq \bar v-\eta$, this follows from weakly undominated bidding. If
$X>\bar v-\eta$, bidding $a$ wins whenever the native bidder's value is below
$a$, and so gives payoff at least
$(X-a)F(a)>(\bar v-\eta-a)F(a)>\eta$. By contrast, any voluntary bid
$b>\bar v-\eta$ gives payoff at most $X-b<\eta$. Hence no bid above
$\bar v-\eta$ can be optimal. Thus $\alpha_n(X)\leq \bar v-\eta$ for all
$X$ and all $n$.

It follows that, for every $b\geq \bar v-\eta$,
\begin{equation*}
b_R(X,Y)\leq b
\quad\Longleftrightarrow\quad
Y\leq b .
\end{equation*}
Near the top of the support, relay bids are therefore governed by the English
floor $Y$, not by the voluntary bid $\alpha_n(X)$.

We now give the native deviator a feasible strategy. Write an upper-tail
native type as
\begin{equation*}
z=F^{-1}\left(1-\frac{t}{n}\right).
\end{equation*}
For this type, let the deviator bid
\begin{equation*}
b_n(t)=F^{-1}\left(1-\frac{x(t)}{n}\right),
\end{equation*}
where $x(t)$ is defined by
\begin{equation*}
t=x(t)-1-\frac{1}{x(t)}.
\end{equation*}
Equivalently,
\begin{equation*}
x(t)=\frac{1+t+\sqrt{(1+t)^2+4}}{2}.
\end{equation*}
Since $x(t)>t$, this bid is below the type's value. Moreover, for each fixed
$t$, $b_n(t)\to\bar v$, so for large $n$ the bid lies in the region where the
voluntary relay bid does not bind.

The native bidder therefore wins exactly when $Y\leq b_n(t)$. Since the
relay contains $n-1$ iid values,
\begin{equation*}
\Pr(Y\leq b_n(t))
=
\left(1-\frac{x(t)}{n}\right)^{n-1}
+
(n-1)\frac{x(t)}{n}
\left(1-\frac{x(t)}{n}\right)^{n-2},
\end{equation*}
which converges to $e^{-x(t)}(1+x(t))$. Also,
\begin{equation*}
F^{-1}\left(1-\frac{t}{n}\right)
-
F^{-1}\left(1-\frac{x(t)}{n}\right)
=
\frac{x(t)-t}{nf(\overline v)}
+
o\left(\frac{1}{n}\right).
\end{equation*}

Integrating only over types $t\in[0,T]$, and using $dF=dt/n$, yields
\begin{equation*}
\liminf_{n\to\infty} n^2f(\overline v)\,U_{\mathrm{out}}^n(F)
\geq
\int_0^T
(x(t)-t)e^{-x(t)}(1+x(t))\,dt .
\end{equation*}
Letting $T\to\infty$ and using
\begin{equation*}
t=x-1-\frac{1}{x},
\qquad
x-t=1+\frac{1}{x},
\qquad
dt=\left(1+\frac{1}{x^2}\right)dx,
\end{equation*}
we obtain
\begin{equation*}
\liminf_{n\to\infty} n^2f(\overline v)\,U_{\mathrm{out}}^n(F)
\geq
\rho,
\end{equation*}
where
\begin{equation*}
\rho
=
\int_{\varphi}^{\infty}
e^{-x}(1+x)
\left(1+\frac{1}{x}\right)
\left(1+\frac{1}{x^2}\right)\,dx,
\qquad
\varphi=\frac{1+\sqrt 5}{2}.
\end{equation*}
The constant is strictly larger than one. Indeed, the integrand equals
$e^{-x}(x+2+2/x+2/x^2+1/x^3)$, and hence
\begin{equation*}
\rho
>
e^{-\varphi}(\varphi+3)
+
\frac{13}{8}(e^{-\varphi}-e^{-2})
>
1 .
\end{equation*}
Combining the two estimates,
\begin{equation*}
n^2f(\overline v)\,U_E^n(F)\to 1,
\qquad
\liminf_{n\to\infty} n^2f(\overline v)\,U_{\mathrm{out}}^n(F)\geq \rho>1.
\end{equation*}
Therefore $U_{\mathrm{out}}^n(F)>U_E^n(F)$ for all sufficiently large $n$.
A builder in the all-English profile then has a profitable deviation to the
native first-price channel. Hence the all-English profile is not a
first-stage equilibrium for all sufficiently large $n$.
\end{proof}

\subsection{Proof of Proposition~\ref{prop:commitment_show}}
Observe that both auction formats are non-wasteful: for every bid profile, the
block is allocated to some bidder.

For a valuation profile $v$, let $p_i^S(v)$ and $x_i^S(v)$ denote,
respectively, the payment and allocation probability of bidder $i$ in a
Bayes--Nash equilibrium of $\texttt{Show}(S)$. By the revelation principle and
Myerson's payment identity,
\begin{equation*}
\mathbb{E}\!\left[\sum_{i=1}^n p_i^S(v)\right]
=
\mathbb{E}\!\left[\sum_{i=1}^n \phi(v_i)x_i^S(v)\right]
-
\sum_{i=1}^n u_i^S(\underline v),
\end{equation*}
where $u_i^S(\underline v)$ is bidder $i$'s interim utility at the lowest type.
Since every bidder has outside option $0$, we have $u_i^S(\underline v)\geq0$.
Therefore,
\begin{equation*}
\mathrm{Rev}(\texttt{Show}(S))
\leq
\mathbb{E}\!\left[\sum_{i=1}^n \phi(v_i)x_i^S(v)\right].
\end{equation*}

Since $F$ is regular and bidders are identically distributed, the highest-value
bidder also has the highest virtual value. Moreover, $\texttt{Show}(S)$ is
non-wasteful, so $\sum_i x_i^S(v)=1$ for every $v$. Hence, pointwise,
\begin{equation*}
\sum_{i=1}^n \phi(v_i)x_i^S(v)
\leq
\sum_{i=1}^n
\phi(v_i)\mathbf{1}\!\left\{v_i>\max_{j\neq i}v_j\right\},
\end{equation*}
where ties have probability zero because $F$ is atomless. Taking expectations,
\begin{equation*}
\mathrm{Rev}(\texttt{Show}(S))
\leq
\mathbb{E}\!\left[
\sum_{i=1}^n
\phi(v_i)\mathbf{1}\!\left\{v_i>\max_{j\neq i}v_j\right\}
\right].
\end{equation*}
The right-hand side is the expected virtual surplus of the efficient
allocation. The symmetric first-price auction allocates efficiently and gives
the lowest type utility $0$, so by Myerson's payment identity, equivalently by
revenue equivalence,
\begin{equation*}
\mathrm{Rev}(\mathrm{FP})
=
\mathbb{E}\!\left[
\sum_{i=1}^n
\phi(v_i)\mathbf{1}\!\left\{v_i>\max_{j\neq i}v_j\right\}
\right].
\end{equation*}
Combining the previous two displays gives the result.

\subsection{Proof of Proposition~\ref{cor:open_relay_no_leakage}}
\label{app:open_relay_no_leakage}
\begin{proof}
Let $S$ denote the set of low-latency builders and let
$A=[n]\setminus S$ denote the set of high-latency builders. Fix a
high-latency builder $i\in A$ with value $v_i$, and fix a bid $b\leq v_i$.
Compare two ways of submitting this same bid. In the first, builder $i$ sends
$b$ through the open relay. In the second, builder $i$ sends $b$ through the
sealed first-price channel.

By assumption, the proposer treats the bid in the same way in both cases: the
bid is compared against all other bids, and if it wins, builder $i$ pays $b$.
The only difference is informational. If $b$ is sent through the open relay,
the low-latency builders observe it before submitting their final bids. If
$b$ is sent through the sealed channel, they do not.

We first show that the sealed submission weakly dominates the relay
submission. Consider any realization of the other builders' values and bids.
If the open-relay bid $b$ wins, then no low-latency builder can both observe
$b$ and have value strictly above $b$. Indeed, any such builder could bid
slightly above $b$, or bid $b$ when ties favor low-latency builders, and win
with positive surplus. Hence, whenever the open-relay bid wins, the same bid
submitted through the sealed channel also wins: the bid is still submitted to
the proposer, the payment is the same, and hiding the bid only removes the
low-latency builders' ability to condition their final bids on it.

The comparison is strict on the states in which some low-latency builder would
react to the revealed bid. More precisely, suppose that a low-latency builder
has value above $b$, but would not have bid above $b$ without observing this
particular relay bid. If $b$ is sent through the sealed channel, builder $i$
can still win on such a state. If $b$ is sent through the open relay, the
low-latency builder observes $b$ and can profitably outbid it. Thus the relay
submission loses in a state in which the sealed submission may win.

Under the maintained atomless full-support value distribution, and under the
usual first-price shading in the sealed channel, such states occur with
positive probability for almost every type whose bid has a positive chance of
winning. Therefore a high-latency builder strictly prefers, for almost every
such type, to submit through the sealed channel rather than through the open
relay.

It follows that high-latency builders do not send their bids through the open
relay. Once this happens, the relay no longer reveals the bids that
low-latency builders need to beat. Low-latency builders may still observe the
relay, or even submit messages through it, but the relay gives them no
information about the sealed high-latency bids. Thus the last-look advantage
disappears. The remaining payoff-relevant bidding environment is the standard
sealed first-price auction, so the allocation and payments coincide with the
standard first-price equilibrium outcome.
\end{proof}

\end{document}